\documentclass[12pt]{article}
\author{D.S. Malyshev\thanks{Laboratory of Algorithms and Technologies for Network Analysis,
National Research University Higher School of Economics, 136
Rodionova str., Nizhny Novgorod 603093, Russia; Department of
Applied Mathematics and Informatics, National Research University
Higher School of Economics, 25/12 Bolshaya Pecherskaya str., Nizhny
Novgorod, 603155, Russia; Department of Mathematical Logic and
Higher Algebra, Lobachevsky State University of Nizhny Novgorod, 23
Gagarina av., Nizhny Novgorod, 603950, Russia; Tel: 8 (831)
436-13-97; Email: dsmalyshev@rambler.ru}$~^{,}$\thanks{The research
is partially supported by Russian Foundation for Basic Research,
grant No 12-01-00749-a; by Federal Target Program "Researh and
educational specialists of innovative Russia", state contract No
14.Â37.21.0393; by LATNA laboratory, National Research University
Higher School of Economics, RF government grant, ag.
11.G34.31.00357; by RF President grant MK-1148.2013.1. This study
was carried out within "The National Research University Higher
School of Economics' Academic Fund Program" in 2013-2014, research
grant No. 12-01-0035.}}
\date{}
\title{The coloring problem for classes with two small obstructions}

\usepackage{amsmath,amsfonts,amssymb}
\usepackage{latexsym}

\begin{document}
\maketitle
\newtheorem{theorem}{Theorem}
\newtheorem{corollary}{Corollary}
\newtheorem{conjecture}{Conjecture}
\newtheorem{lemma}{Lemma}
\newenvironment{proof}[1][Proof]{\textbf{#1.} }{\ \rule{0.5em}{0.5em}}

\begin{abstract}

The coloring problem is studied in the paper for graph classes
defined by two small forbidden induced subgraphs. We prove some
sufficient conditions for effective solvability of the problem in
such classes. As their corollary we determine the computational
complexity for all sets of two connected forbidden induced subgraphs
with at most five vertices except 13 explicitly enumerated cases.\\

\end{abstract}

\section{Introduction}

The coloring problem is one of classical problems on graphs. Its
formulation is as follows. A \emph{coloring} is an arbitrary mapping
of colors to vertices of some graph. A graph coloring is called
\emph{proper} if any neighbors are colored in different colors. The
\emph{chromatic number of a graph} $G$ (denoted by $\chi(G)$) is the
minimal number of colors in proper colorings of $G$. The \emph
{coloring problem} for a given graph is to find its chromatic
number. The \emph{vertex $k$-colorability problem} is to verify
whether vertices of a given graph can be colored with at most $k$
colors. The \emph{edge $k$-colorability problem} is formulated by
analogy.

A graph $H$ is called an \emph{induced subgraph} of $G$ if $H$ is
obtained from $G$ by deletions of its vertices. The \emph{induced
subgraph relation} is denoted by $\subseteq_i$. In other words,
$H\subseteq_iG$ if $H$ is an induced subgraph of $G$. A \emph{class}
is a set of simple unlabeled graphs. A class of graphs is called
\emph{hereditary} if it is closed under deletions of vertices. It is
well known that any hereditary (and only hereditary) graph class
${\mathcal X}$ can be defined by a set of its forbidden induced
subgraphs ${\mathcal S}$. We write ${\mathcal X}=Free({\mathcal S})$
in this case. If a hereditary class can be defined by a finite set
of forbidden induced subgraphs, then it is called \emph{finitely
defined}.

The coloring problem for $G$-free graphs is polynomial-time solvable
iff $G\subseteq_iP_4$ or $G\subseteq_iP_3\oplus K_1$ (Kral' et al.
2002). A study of forbidden pairs was also initialized in the paper.
When we forbid two induced subgraphs, the situation becomes more
difficult than in the monogenic case. Here only partial results are
known (Kral' et al. 2002; Brandst\"{a}dt et al. 2002; Brandst\"{a}dt
et al. 2006; Schindl 2005; Golovach, Paulusma and Song 2011;
Dabrowski et al. 2012; Golovach and Paulusma 2013). The next
statement is a survey of such achievements (Golovach and Paulusma
2013).

\begin{theorem}
Let $G_1$ and $G_2$ be two fixed graphs. The coloring problem is
NP-complete for $Free(\{G_1,G_2\})$ if:

\begin{itemize}
\item $C_p\subseteq_iG_1$ for some $p\geq 3$ and $C_q\subseteq_iG_2$ for some $q\geq 3$
\item $K_{1,3}\subseteq_iG_1$ and $K_{1,3}\subseteq_iG_2$
\item $K_{1,3}\subseteq_iG_1$ and either $K_4\subseteq_iG_2$ or
$K_4-e\subseteq_iG_2$ $($or vice versa$)$
\item $K_{1,3}\subseteq_iG_1$ and $C_p\subseteq_iG_2$ for some $p\geq 4$
$($or vice versa$)$
\item $G_1$ and $G_2$ contain a spanning subgraph of $2K_2$ as an
induced subgraph
\item $C_3\subseteq_iG_1$ and $K_{1,p}\subseteq_iG_2$ for some $p\geq 5$ $($or vice
versa$)$
\item $C_3\subseteq_iG_1$ and $P_{164}\subseteq_iG_2$ $($or vice
versa$)$
\item $C_p\subseteq_iG_1$ for $p\geq 5$ and $G_2$ contains a spanning subgraph of $2K_2$ as an
induced subgraph $($or vice versa$)$
\item either $C_p\oplus K_1\subseteq_iG_1$ for $p\in \{3,4\}$ or $\overline{C_q}\subseteq_iG_1$ for $q\geq 6$ and
$G_2$ contains a spanning subgraph of $2K_2$ as an induced subgraph
$($or vice versa$)$
\end{itemize}

It is polynomial-time solvable for $Free(\{G_1,G_2\})$ if:
\begin{itemize}
\item $G_1$ and $G_2$ are induced subgraphs of $P_4$ or $P_3\oplus
K_1$
\item $G_1\subseteq_iK_{1,3}$ and $G_2\subseteq_iC_3\oplus K_1$ $($or vice
versa$)$
\item $G_1\subseteq_ipaw$ and $G_2\neq K_{1,5}$ is a forest with at most six vertices $($or
vice versa$)$
\item $G_1\subseteq_ipaw$ and either $G_2\subseteq_i pK_2$ or $G_2\subseteq_i P_5\oplus pK_1$
for some $p\geq 1$ $($or vice versa$)$
\item $G_1\subseteq_iK_p$ for $p\geq 3$ and either $G_2\subseteq_i qK_2$ or $G_2\subseteq_i P_5\oplus qK_1$
for some $q\geq 1$ $($or vice versa$)$
\item $G_1\subseteq_igem$ and either $G_2\subseteq_i P_4\oplus K_1$ or $G_2\subseteq_i
P_5$ $($or vice versa$)$
\item $G_1\subseteq_i{\overline P_5}$ and either $G_2\subseteq_i P_4\oplus K_1$ or $G_2\subseteq_i
2K_2$ $($or vice versa$)$
\end{itemize}
\end{theorem}

In the present article we prove some sufficient conditions for
NP-completeness and polynomial-time solvability of the coloring
problem for $\{G_1,G_2\}$-free graphs. They add new information
about its complexity for some cases that Theorem 1 does not cover.
For instance, the problem is appeared to be NP-complete for
$\{K_{1,4},bull\}$-free graphs, but it is polynomial-time solvable
for $Free(\{K_{1,3},P_5\}), Free(\{K_{1,3},hammer\}),
Free(\{P_5,C_4\})$. The complexity was earlier open for these four
cases. As a corollary of the conditions we determine the complexity
for all sets $\{G_1,G_2\}$ of connected graphs with at most five
vertices except 13 listed cases.

\section{Notation}

As usual, $P_n,C_n,K_n,O_n$ and $K_{p,q}$ stand respectively for the
simple path with $n$ vertices, the chordless cycle with $n$
vertices, the complete graph with $n$ vertices, the empty graph with
$n$ vertices and the complete bipartite graph with $p$ vertices in
the first part and $q$ vertices in second. The graph $K_n-e$ is
obtained by deleting an arbitrary edge in $K_n$. The graph
\emph{paw} is obtained from $K_{1,3}$ by adding a new edge incident
to its vertices of degree two. The graphs
$fork,gem,hammer,bull,butterfly$ have the vertex set
$\{x_1,x_2,x_3,x_4,x_5\}$. The edge set for \emph{fork} is
$\{(x_1,x_2),(x_1,x_3),(x_1,x_4),(x_4,x_5)\}$, for \emph{gem} is
$\{(x_1,x_2),\\
(x_1,x_3),(x_1,x_4), (x_1,x_5),(x_2,x_3),(x_3,x_4),(x_4,x_5)\}$, for
\emph{hammer} is
$\{(x_1,x_2),\\
(x_1,x_3),(x_2,x_3),(x_1,x_4),(x_4,x_5)\}$, for \emph{bull} is
$\{(x_1,x_2),(x_1,x_3), (x_2,x_3),(x_1,x_4),\\
(x_2,x_5)\}$, for
\emph{butterfly} is
$\{(x_1,x_2),(x_1,x_3),(x_2,x_3),(x_1,x_4),(x_1,x_5),(x_4,x_5)\}$.

The \emph{complemental graph} of $G$ (denoted by $\overline{G}$) is
a graph on the same set of vertices and two vertices of
$\overline{G}$ are adjacent if and only if they are not adjacent in
$G$. The sum $G_1\oplus G_2$ is the disjoint union of $G_1$ and
$G_2$. The disjoint union of $k$ copies of a graph $G$ is denoted by
$kG$. For a graph $G$ and a set $V\subseteq V(G)$ the formula
$G\setminus V$ denotes the subgraph of $G$ obtained by deleting all
vertices in $V$.

All graph notions and properties that are not formulated in this
paper can be found in the textbooks (Bondy and Murty 2008; Distel
2010).

\section{Boundary graph classes}

The notion of a boundary graph class is a helpful tool for analysis
of the computational complexity of graph problems in the family of
hereditary graph classes. This notion was originally introduced by
Alekseev for the independent set problem (Alekseev 2004). It was
applied for the dominating set problem later (Alekseev, Korobitsyn
and Lozin 2004). A study of boundary graph classes for some graph
problems was extended in the paper (Alekseev et al. 2007), where the
notion was formulated in its most general form. Let us give the
necessary definitions.

Let $\Pi$ be an NP-complete graph problem. A hereditary graph class
is called {\it $\Pi$-easy} if $\Pi$ is polynomial-time solvable for
its graphs. If the problem $\Pi$ is NP-complete for graphs in a
hereditary class, then this class is called {\it $\Pi$-hard}. A
class of graphs is said to be \emph{$\Pi$-limit} if this class is
the limit of an infinite monotonously decreasing chain of $\Pi$-hard
classes. In other words, ${\mathcal X}$ is $\Pi$-limit if there is
an infinite sequence ${\mathcal X_1}\supseteq {\mathcal X_1}
\supseteq \ldots$ of $\Pi$-hard classes, such that ${\mathcal
X}=\bigcap\limits_{k=1}^{\infty}{\mathcal X_k}$. A minimal under
inclusion $\Pi$-limit class is called \emph{$\Pi$-boundary}.

The following theorem certifies the significance of the boundary
class notion (Alekseev 2004).

\begin{theorem}
A finitely defined class is $\Pi$-hard iff it contains some
$\Pi$-boundary class.
\end{theorem}

The theorem shows that knowledge of all $\Pi$-boundary classes leads
to a complete classification of finitely defined graph classes with
respect to the complexity of $\Pi$. Two concrete classes of graphs
are known to be boundary for several graph problems. First of them
is ${\mathcal S}$. It is constituted by all forests with at most
three leaves in each connected component. The second one is
${\mathcal T}$, which is a set of the line graphs of graphs in
${\mathcal S}$. The paper (Alekseev et al. 2007) is a good survey
about graph problems, for which either ${\mathcal S}$ or ${\mathcal
T}$ is boundary.

Some classes are known to be limit and boundary for the coloring
problem. The set of all forests (denoted by ${\mathcal F}$) and the
set of line graphs of forests with degrees at most three are limit
classes for it (Lozin and Kaminski 2007). The last class we will
denote by ${\mathcal T'}$. The set $co({\mathcal
T})=\{G:~\overline{G}\in {\mathcal T}\}$ is a boundary class for the
problem (Malyshev 2012(a)). The set of boundary graph classes for
the coloring problem is continuous (Korpeilainen et al. 2011). Some
continuous sets of boundary classes for the vertex $k$-colorability
and the edge $k$-colorability problems are known for any fixed
$k\geq 3$ (Malyshev 2012(a); Malyshev 2012(b)).

\section{NP-completeness of the coloring problem for
$\{K_{1,4},bull\}$-free graphs}

The listed above results on limit and boundary classes for the
coloring problem together with Theorem 1 allow to prove
NP-completeness of the problem for some finitely defined classes.
Namely, if ${\mathcal Y}$ is a finite set of graphs and no one graph
in ${\mathcal Y}$ belongs to a class in  $\{{\mathcal F}, {\mathcal
T'}, co({\mathcal T})\}$, then the problem is NP-complete for
$Free({\mathcal Y})$. But this idea can not be applied to
$Free(\{K_{1,4}, bull\})$, because $K_{1,4}\in {\mathcal F}, bull\in
{\mathcal T'},bull\in co({\mathcal T})$. Nevertheless, the coloring
problem is NP-complete for it. To show this we use the operation
with a graph called the diamond implantation.

Let $G$ be a graph and $x$ be its nonpendant vertex. Applying the
{\it diamond implantation} to $x$ implies:

\begin{itemize}
\item an arbitrary splitting of the neighborhood of $x$ into two nonempty parts $A$ and
$B$

\item deletion of $x$ and addition of new vertices $y_1,y_2,y_3,y_4$

\item addition of all edges of the kind $(y_1,a), a\in A$ and of the
kind $(y_4,b),b\in B$

\item addition of the edges
$(y_1,y_2),(y_1,y_3),(y_2,y_3),(y_2,y_4),(y_3,y_4)$

\end{itemize}

Clearly that for any graph and any its nonleaf vertex applying the
diamond implantation preserves vertex 3-colorability. This property
and the paper (Kochol, Lozin and Randerath 2003) give the key idea
of the proof of Lemma 1.

\begin{lemma}
The vertex $3$-colorability problem $($hence, the coloring
problem$)$ is NP-complete for the class $Free(\{K_{1,4}, bull\})$.
\end{lemma}

\begin{proof} The vertex 3-colorability problem is known to be NP-complete
for triangle-free graphs with degrees at most four (Maffray and
Preissmann 1996). Let us consider connected such a graph with at
least two vertices. We will sequentially apply the described above
operation to its vertices with edgeless neighborhoods. In other
words, if $H$ is a current graph, then it is applied to an arbitrary
vertex of $H$ that does not belong to any triangle. The sets $A$ and
$B$ are arbitrarily formed with the condition $||A|-|B||\leq 1$. The
whole process is finite, because the number of its steps is no more
than the quantity of vertices in the initial graph. It is easy to
see that the resulted graph belongs to $Free(\{K_{1,4}, bull\})$.
Thus, the vertex 3-colorability problem for triangle-free graphs
with degrees at most four is polynomially reduced to the same
problem for graphs in $Free(\{K_{1,4}, bull\})$. Hence, it is
NP-complete for $Free(\{K_{1,4}, bull\})$.\end{proof}

\section{Some structural results on graphs in some classes defined by
small obstructions}

For a hereditary class ${\mathcal X}$ and a number $k$ the formula
$[{\mathcal X}]_k$ is a set of graphs, for which one can delete at
most $k$ vertices that a result belongs to ${\mathcal X}$.

\begin{lemma}
If some connected graph $G\in Free(\{K_{1,3},P_5\})$ contains an
induced cycle $C$ of length at least four, then $G\in
[Free(\{O_3\})]_5$.
\end{lemma}

\begin{proof} Length of $C$ is equal to either four or five. We will show
first that $C$ dominates all vertices of $G$. Let there is a vertex
of $G$ that does not belong to $C$ and adjacent to no one vertex of
the cycle. Then due to the connectivity of $G$ there are vertices
$x,y\in V(G)\setminus V(C)$, such that $(x,y)\in E(G)$, $x$ is not
adjacent to any vertex of $C$ and $y$ is adjacent to at least one
vertex of $C$. Since $G\in Free(\{K_{1,3}\})$, then $y$ is adjacent
to exactly two vertices of $C$. The vertices $x,y$ and some three
consecutive vertices of the cycle (one of which is adjacent to $y$)
induce the subgraph $P_5$. Thus, $C$ really dominates all vertices
of $G$.

We will show that the graph $G\setminus V(C)$ does not contain three
pairwise nonadjacent vertices. This fact implies the validity of
Lemma 2. Assume that $G\setminus V(C)$ has a set $ V $ of three
pairwise nonadjacent vertices. Since $G$ is $K_{1,3}$-free, then the
intersection of the neighborhood of each vertex in $V$ with $V(C)$
is a set of at least two (three for $C=C_5$) consecutive vertices of
$C$.

Let us consider the case $C=C_5$. No one vertex of $V$ can be
adjacent to all vertices of $C$, since otherwise some vertex of $C$
is adjacent to all vertices of $V$ (and $ G $ contains $K_{1,3}$ as
an induced subgraph). One can assume that no one vertex in $V$ is
adjacent to exactly four vertices of the cycle $ C $, since in this
case the graph $G$ contains the induced cycle $C_4$ and the case
$C=C_4$ will be considered later. Therefore, we can consider only
the situation, where each vertex of $ V $ is adjacent to three
consecutive vertices of $ C $ and the corresponding sets of three
consecutive vertices are distinct (otherwise $G$ contains $K_{1,3}$
as an induced subgraph). Then, some two vertices of $V$ and some
three vertices of $C$ induce $P_5$. So, if $C=C_5$, then we have a
contradiction.

Now we consider the case $C=C_4$. It is easy to verify that avoiding
induced $K_{1,3}$ in $G$ leads to only the following situations:

\begin{itemize}
\item one vertex of $V$ is adjacent to all
vertices of $C$ and the other two vertices of $V$ are adjacent to
disjoint pairs of its consecutive vertices

\item one vertex of $V$ is adjacent to two consecutive vertices of $C$ and
each of the other two vertices is adjacent to three consecutive
vertices of $C$, they have two common neighbors in $C$ and the first
vertex has only one common neighbor in $C$ with each of them

\item each of two vertices of $V$ is adjacent to two consecutive vertices of $C$,
the third one is adjacent to three consecutive vertices of $C$ and
any two vertices of $V$ have only one common neighbor in $C$
\end{itemize}

The graph $ G $ contains $ P_5 $ as an induced subgraph in all three
cases. We come to a contradiction. Thus, the initial assumption was
false. \end{proof}

\begin{lemma}
If some connected graph $G\in
Free(\{K_{1,3},hammer\})$ contains an induced cycle $C_n~(n\geq 7)$,
then $G$ is isomorphic to $C_n$.
\end{lemma}

\begin{proof} Assume opposite, i.e. there is a vertex $x\in
V(G)\setminus V(C_n)$. One can easily show that the vertex $x$ is
adjacent to at least one vertex of $C_n$. It is easy to verify that
the set of $x's$ neighbors in $C_n$ is constituted either by two,
three or four consecutive vertices or by two pairs of consecutive
vertices (otherwise $G\not \in Free(\{K_{1,3}\})$). In both
situations the graph $G$ contains $hammer$ as an induced subgraph.
Hence, the assumption was false.\end{proof}

\begin{lemma}
If some connected graph $G\in
Free(\{K_{1,3},hammer\})$ contains $C_6$ as an induced subgraph,
then $G\setminus V(C_6)$ is the disjoint union of at most three
cliques.
\end{lemma}

\begin{proof} Let us consider the set $V=V(G)\setminus V(C_6)$. It is
easy to verify that the intersection of the neighborhood of each
vertex in $V$ with $C_6$ induces in $G$ the subgraph $2K_2$. Let us
consider now two arbitrary vertices in $V$. If they are adjacent,
then they have in $C_6$ the same sets of neighbors and if they are
not adjacent, then the mentioned sets are distinct. This implies
that $V$ does not contain four pairwise nonadjacent vertices. Thus,
$G\setminus V(C_6)$ is the disjoint union of at most three cliques.
\end{proof}

\begin{lemma}
For any connected graph $G\in Free(\{K_{1,3},hammer\})$ at least one
of the following properties is true:

\begin{itemize}
\item $G$ is a simple cycle

\item $G$ contains the induced subgraph $C_6$

\item $G$ has a pendant vertex

\item $G$ belongs to the class $Free(\{P_5\})$

\item $G$ belongs to the class $[Free(\{O_3\})]_5$
\end{itemize}
\end{lemma}

\begin{proof} Assume that $G\not \in Free(\{P_5\})$. Let us consider
an induced path $P_n$ of $G$ having the maximal length. Clearly,
$n\geq 5$. Let us consider an arbitrary end of this path. One can
assume that it is adjacent to some vertex $x\in V(G)\setminus
V(P_n)$, otherwise $G$ contains a pendant vertex. By the maximality
of $P_n$ the vertex $x$ is adjacent to at least two vertices of the
path. One can consider that $x$ is adjacent to at least one interior
vertex of $P_n$, otherwise $G$ is a simple cycle (by Lemma 3) or it
contains $C_6$ as an induced subgraph.

Let $n>5$. To avoid induced $K_{1,3}$ the vertex $x$ must be
adjacent to three or four consecutive vertices of $P_n$ or to two
end its vertices or to three vertices of $P_n$ that induce the
subgraph $K_2\oplus K_1$ in $G$ or to four vertices inducing $2K_2$.
The graph $G$ contains $hammer$ as an induced subgraph in all these
situations.

Let $n=5$ now. One can assume that the graph $G\setminus V(P_5)$ has
three pairwise nonadjacent vertices (otherwise $G\in
[Free(\{O_3\})]_5$). It is easy to check that any of these three
vertices must be adjacent to either three central vertices of $P_5$
or to all its vertices, except central or to the first, the third
and the fourth vertices of $P_5$ (counting from some of the $P_5$'s
ends) or to the first and the last its vertices. The graph $G$
contains $C_6$ as an induced subgraph in the last case. Hence, we
can consider that no one among the three vertices is adjacent to
only the ends of $P_5$. If one of the three vertices is adjacent to
the first, the third and the fourth vertices of $P_5$ and other of
these vertices is adjacent to the second, the third and the fifth
ones, then $G$ contains induced $C_6$. Therefore, one can assume
that there are no such two vertices. Either the second or the fourth
vertex of $P_5$ is adjacent to the three vertices and, hence, $G$ is
not $K_{1,3}$-free. Thus, the initial assumption was false.
\end{proof}

\section{On formulae connecting the chromatic numbers of a graph and of its
induced subgraphs}

The following statement is obvious.

\begin{lemma}
If $G$ is a connected graph with at most three vertices
and a pendant vertex $v$, then $\chi(G\setminus \{v\})=\chi(G)$.
\end{lemma}

\begin{lemma}
Let $G$ be a connected graph in $Free(\{P_5,C_4\})$ that contains
induced $C_5$. Let $V_1$ be the set its vertices that adjacent to
all vertices of $C_5$, $V_2$ be the set of vertices in $G$ that have
three neighbors in $C_5$, $G_1$ and $G_2$ be the subgraphs of $G$,
induced by $V(G)\setminus (V_1\cup V_2\cup V(C_5))$ and $V_1\cup
V_2\cup V(C_5)$ correspondingly. Then, $G_2$ is $O_3$-free and the
relation $\chi(G)=\max(\chi(G_1),\chi(G_2))$ holds.
\end{lemma}

\begin{proof} Any vertex outside $C_5$ that adjacent to at least one vertex
of the cycle must be adjacent to either all vertices of the cycle or
to three consecutive its vertices. It is easy to verify taking into
account that $G$ is $\{P_5,C_4\}$-free. Therefore, any such a vertex
belongs to either $V_1$ or $V_2$. Each vertex in $V_2$ has no a
neighbor outside $V(C_5)\cup V_1\cup V_2$ (since $G\in
Free(\{P_5\})$). As $G$ is $C_4$-free, then any vertex in $V_1$ is
adjacent to every vertex in $V_1\cup V_2\cup V(C_5)$ except itself.
It is easy to verify that $G_2$ is $O_3$-free.

The inequality $\chi(G)\geq \max(\chi(G_1),\chi(G_2))$ is obvious.
We will show that $G$ can be colored with
$\max(\chi(G_1),\chi(G_2))$ colors. Let $c_1$ and $c_2$ be optimal
colorings of $G_1$ and $G_2$ correspondingly. If $\chi(G_1)\geq
\chi(G_2)$, then $c_1$ has $\chi(G_1)-|V_1|\geq\chi(G_2)-|V_1|\geq
0$ colors that do not appear in $V_1$. Hence, $c_1$ can be extended
to a proper coloring of $G$ with $\chi(G_1)$ colors by coloring
$G_2\setminus V_1$ with $\chi(G_2)-|V_1|$ colors of the mentioned
type. By the same reasons $c_2$ is extendable to a proper coloring
of $G$ with $\chi(G_2)$ colors when $\chi(G_2)\geq \chi(G_1)$.
\end{proof}

\section{Some results on polynomial-time solvability of the coloring problem}

\begin{lemma}
Let ${\mathcal X}$ be an easy case for the coloring problem and for
some number $p$ the inclusion ${\mathcal X}\subseteq Free(\{O_p\})$
holds. Then, for any fixed $q$ this problem is polynomial-time
solvable in the class $[{\mathcal X}]_q$.
\end{lemma}

\begin{proof} Let $G$ be a graph in $[{\mathcal X}]_q$. Deleting some set
$V$ ($|V|\leq q$) of its vertices leads to a graph in ${\mathcal
X}$. We will consider all partial proper colorings of $G$ with at
most $|V|$ color classes, in which every vertex of $V$ is colored.
Obviously, any such a coloring has at most $(p-1)q$ colored
vertices. Hence, the quantity of the colorings is bounded by a
polynomial on $|V(G)|$. For any considered partial coloring deleting
all colored vertices leads to a graph in ${\mathcal X}$ and its
chromatic number is computed in polynomial time. For every our
partial coloring we will find the sum of the number of used colors
and the chromatic number of the subgraph induced by the set of
uncolored vertices. Minimal among these sums is equal to $\chi(G)$.
Thus, $\chi(G)$ is computed in polynomial time.\end{proof}

A graph is called {\it chordal} if it does not contain induced
cycles with four and more vertices. The coloring problem is known to
be polynomial-time solvable for chordal graphs (Golumbic 1980).

\begin{lemma}
The classes $Free(\{K_{1,3},P_5\}), Free(\{K_{1,3},hammer\}),
Free(\{P_5,C_4\})$ are easy for the coloring problem.
\end{lemma}

\begin{proof} We will show that for every considered class the coloring problem is
polynomially reduced to the same problem for chordal graphs. This
fact implies the lemma. The problem is polynomial-time solvable in
$Free(\{O_3\})$, since it is equivalent to the matching problem. The
reduction for $\{K_{1,3},P_5\}$-free graphs follows from this
observation, Lemma 2 and Lemma 8.

Let $G$ be a graph in $Free(\{K_{1,3},hammer\})$ containing the
induced subgraph $C_6$. By Lemma 4, deleting vertices of this cycle
leads to a chordal $O_4$-free graph. Hence, by Lemma 8, $\chi(G)$ is
computed in polynomial time. Thus, by Lemma 5 and Lemma 6 the
coloring problem for the class is polynomially reduced to the same
problem for graphs in $Free(\{K_{1,3},P_5\})\cup [Free(\{O_3\})]_5$.
Hence, it is reduced to chordal graphs.

Let $G$ be a connected graph in $Free(\{P_5,C_4\})$ that is not
chordal. Hence, $G$ contains induced $C_5$. The graphs $G_1$ and
$G_2$ defined in the formulation of Lemma 7 are constructed in
polynomial time. Moreover, $G_2$ is $O_3$-free and
$|V(G)|-|V(G_1)|\geq 5$. Therefore, by Lemma 7 the considered
problem for $\{P_5,C_4\}$-free graphs is also polynomially reduced
to the same problem for chordal graphs.\end{proof}

\section{The main result and its corollaries}

The following theorem is the main result of the paper.

\begin{theorem}
Let $H_1$ and $H_2$ be some graphs. If there is a class ${\mathcal
Y}\in \{{\mathcal F},{\mathcal T'},co({\mathcal T})\}$ with either
$H_1,H_2\not \in {\mathcal Y}$ or $K_{1,4}\subseteq_iH_1$ and
$bull\subseteq_iH_2$ $($or vice versa$)$, then the coloring problem
is NP-complete for $Free(\{H_1,H_2\})$. It is polynomial-time
solvable in the class if at least one of the following properties
holds:

\begin{itemize}
\item $H_1\subseteq_iP_4$ or $H_2 \subseteq_iP_4$
\item $H_1\subseteq_iP_5$ or $H_2 \subseteq_iK_5$ $($or vice
versa$)$
\item $H_1\subseteq_iP_5$ or $H_2 \subseteq_igem$ $($or vice
versa$)$
\item $H_1\subseteq_iP_5$ or $H_2 \subseteq_iC_4$ $($or vice
versa$)$
\item $H_1\subseteq_iP_5$ or $H_2 \subseteq_iK_{1,3}$ $($or vice
versa$)$
\item $H_1\subseteq_iK_{1,4}$ or $H_2 \subseteq_ipaw$ $($or vice
versa$)$
\item $H_1\subseteq_ifork$ or $H_2 \subseteq_ipaw$ $($or vice
versa$)$
\item $H_1\subseteq_iK_{1,3}$ or $H_2 \subseteq_ihammer$ $($or vice
versa$)$
\end{itemize}
\end{theorem}

\begin{proof} Let us recall that the classes ${\mathcal F},{\mathcal
T'},co({\mathcal T})$ are limit for the coloring problem. This fact,
Theorem 1 and Lemma 1 imply the first part of the statement. The set
of $P_4$-free graphs is well known to be an easy case for the
coloring problem (Courcelle and Olariu 2000). The classes
$Free(\{P_5,gem\})$ and $Free(\{P_5,K_5\})$ are easy for the problem
(Brandst\"{a}dt et al. 2002; Golovach and Paulusma 2013). The same
is true for $Free(\{fork,paw\})$ (Golovach and Paulusma 2013)  and
$Free(\{K_{1,4},paw\})$ (Kral' et al. 2001). These facts and Lemma 9
imply the second part of the theorem.\end{proof}

Both parts of Theorem 3 add new information about the complexity of
the coloring problem for some classes. For example, its complexity
status for the classes
$Free(\{K_{1,3},bull\}),Free(\{K_{1,3},P_5\}),
Free(\{K_{1,3},hammer\}), Free(\{P_5,C_4\})$ was open.

Theorem 3 gives the following criterion in the case of connected
$H_1$ and $H_2$ with at most four vertices.

\begin{corollary}
If $H_1$ and $H_2$ are connected graphs with at most four vertices,
then the coloring problem is polynomial-time solvable for
$\{H_1,H_2\}$-free graphs iff either $H_1\subseteq_i P_4$ or
$H_2\subseteq_iP_4$ or $\{H_1,H_2\}=\{K_{1,3},paw\}$ or
$\{H_1,H_2\}=\{K_{1,3},C_3\}$.
\end{corollary}

Theorem 3 can not be applied to some pairs of connected graphs with
at most five vertices. If $\{H_1,H_2\}$ is such a set, then either
$H_1$ or $H_2$ belongs to $\{K_{1,3},fork,K_{1,4},P_5\}$. This
observation helps to enumerate all connected cases with at most five
vertices that the theorem does not cover.

\begin{corollary}

Theorem $3$ does not give the complexity status of the coloring
problem for the following sets of forbidden induced connected
subgraphs $($a number in the brackets shows the quantity of such
kind sets$):$

\begin{itemize}
\item $\{K_{1,3},G\}$, where $G\in \{bull,butterfly\}~(2)$
\item $\{fork,bull\}~(1)$
\item $\{P_5,G\}$, where $G\not \in \{K_5,gem\}$ is an
arbitrary connected five-vertex graph in $co({\mathcal T})~(10)$
\end{itemize}
\end{corollary}

Determining the complexity of the problem for any of the listed
above 13 cases is a challenging research problem.

\end{document}